\documentclass[journal]{IEEEtran}

\usepackage[latin9]{inputenc}
\usepackage{textcomp}
\usepackage{amsmath}
\usepackage{amssymb}
\usepackage{tabularx,ragged2e,booktabs}
\usepackage{epsfig,rotating,setspace,latexsym,amsmath,epsf,amssymb,bm}
\usepackage{cite,graphicx,color,subfigure, amsthm}
\usepackage{subfigure}
\usepackage{multicol}
\usepackage{thmtools}
\usepackage{etoolbox}
\usepackage{enumitem}
\usepackage[autostyle]{csquotes}

\newtheorem{prop}{Proposition}

\usepackage{flushend}
\usepackage{lipsum}

\usepackage[english]{babel}

\ifCLASSINFOpdf

\else

\fi

\begin{document}

\title{Online Caching with Wireless Fronthauling and Delivery in Fog-Aided Networks}

\author{ Seyyed Mohammadreza Azimi
\thanks{Seyyed Mohammadreza Azimi is with the CWiP, Department of Electrical and Computer Engineering, New Jersey Institute of Technology, Newark,
NJ, USA. E-mail: (sa677@njit.edu)}}


\maketitle

\begin{abstract}
Fog  Radio  Access  Network  (F-RAN)  exploits cached contents at edge nodes (ENs)  and  fronthaul connection to the cloud for content delivery. Assuming  dedicated fronthaul links between cloud and each EN, previous works focused on analyses of F-RANs using offline or online caching depending whether the content popularity is  time-invariant or time-variant. Extension has been done for multicast fronthaul link connecting cloud to only two ENs and time-invariant popularity. In contrast, the scope of this work is on the case where multicast fronthaul link connects arbitrary number of ENs to the cloud and content popularity is time-variant. Normalized  Delivery  Time  (NDT) is used as a performance measure and by investigating  proactive online caching,  analytical results reveal that the power scaling of fronthaul transmission sets a limit on the performance of F-RAN.   \end{abstract}

\begin{IEEEkeywords}
 Online  caching, Fog-network, F-RAN, $5$G.
\end{IEEEkeywords}

\IEEEpeerreviewmaketitle

\section{Introduction}
To deliver delay-sensitive multimedia content in $5G$, Fog  Radio  Access  Network  (F-RAN) synergizes edge processing using  cache-aided edge nodes (ENs) and cloud processing using fronthaul connection to the cloud while  Cloud  Radio  Access  Network  (C-RAN) only relies on the latter \cite{OS}.  Assuming time-invariant content, offline caching utilizes separate content placement at ENs and delivery phases in order to serve the users \cite{Avik2}. Time-variability of content necessitate \emph{online} caching with simultaneous cache update and delivery \cite{JSACAzimi}. Offline caching is evaluated in \cite{Avik2} using the performance metric Normalized  Delivery  Time  (NDT) and dedicated fronthaul link between cloud and each EN. Extention to the time-variant popularity is considered in \cite{JSACAzimi}. In contrast, Koh et. al. \cite{Koh} considered time-invariant popularity with wireless multicast fronthaul link connecting cloud to two ENs. The key result was that \textit{pipelined transmission} on the wireless multicast fronthaul channel provides better performance than \textit{coded transmission}. Using pipelined transmission, in this letter, we focus on  \emph{online} caching scenario, the contributions are:     
$(i)$ Deriving NDT for offline caching with arbitrary number of users and ENs as opposed to 2-by-2 case in \cite{Koh}.  
$(ii)$ Proactive online caching with pipelined fronthaul-edge transmission is considered with wireless multicast fronthauling and achievable long-term NDT is derived. This contrasts with \cite{JSACAzimi} in which online caching is utilized for dedicated fronthaul links. 
 
\noindent \textit{Notation}: $X_{[a:b]}=[X_a,X_{a+1},...,X_{b}]$ with $a,b \in \mathbb{N}$ and $a \leq b$ while $X_{[a:b]}=\emptyset$ with $a>b$.  $H(X)$, $h(X)$ and $I(X,Y)$ denote the entropy of $X$, differential entropy of $X$ and the mutual information between $X$ and $Y$.   
   
\section{System Model}\label{sysmo1}
An $M \times K$ F-RAN is considered in which $M$ cache-equipped ENs are connected to the cloud using wireless multicast fronthaul link and serve $K$ users through wireless edge channel. $\mathcal{F}_t$ denotes the set of popular contents at time slot $t$  and it is comprised of $N$ files each of them of size $L$ bits. Each EN can store $\mu NL$ bits and $0 \leq \mu \leq 1$ is defined as the \emph{fractional cache capacity}.
At time slot $t$, users request files from the set $\mathcal{F}_t$. The indices of the uniformly requested files without replacement are denoted by  $d_t=(d_{1,t},...,d_{K,t})$, where $d_{k,t}$ is the requested file by user $k$. 
As proposed in \cite{JSACAzimi}, two-state Markov model is considered to model time varying popularity of contents. At time slot $t$,
with probability $p$, there is an arrival in the set $\mathcal{F}_t$ by randomly replacing a file from the set $\mathcal{F}_{t-1}$; while with probability $1-p$, there is not any arrival  in the popular contents and $\mathcal{F}_t=\mathcal{F}_{t-1}$.
The received signal at the $k$th user in any symbol of the time slot $t$ is
\begin{equation}\label{rcvr}
Y_{k,t}=\sum_{m=1}^{M}H_{k,m,t}X_{m,t} + Z_{k,t},
\end{equation}
where $H_{k,m,t}$ is the channel gain between $m$th EN and $k$th user at time slot $t$; $X_{m,t}$ is the transmitted signal by the $m$th EN; and $Z_{k,t} \sim \mathcal{CN}(0,1)$ is additive noise at $k$th user.  It is assumed that  channel coefficients are independent and identically distributed following a continuous distribution and time-invariant within each transmission interval. 
At time slot $t$, the received signal at $m$th EN is 
\begin{equation}\label{rcvr2}
V_{m,t}=G_{m,t}U_{t} + W_{m,t},
\end{equation}
where $G_{m,t}$ denotes the wireless channel between cloud and $m$th EN, $U_t$ is the transmitted signal by the cloud in channel use $t$ and $W_{m,t}$ is additive noise at $m$th EN. The cloud has a power constraint $T_{F,t}^{-1}\sum_{t=1}^{T_{F,t}} \lvert U_{t} \rvert ^2 \leq P ^r $ with $T_{F,t}$ is the duration (in symbols) of the fronthaul transmission $U_{t}$ in time slot $t$ and $ r \geq 0$ describes the power scaling of the fronthaul transmission as compared to edge transmission. 
At time slot $t$, all the ENs, cloud and users have access to the CSI about the wireless downlink channels $H_t=\{\{H_{k,m,t} \}_{k=1}^{K} \}_{m=1}^{M}$ and the CSI about the wireless fronthaul channel $G_t=\{G_{m,t} \}_{m=1}^{M}$.

 \noindent   \textit{Fronthaul policy}: In time slot $t$, the cloud sends  message $U_{t}$ and it is obtained by the mapping $(d_t,H_{t},G_t,S_{t}) \rightarrow U_{t}$,  with cached contents at ENs  as $S_t=\{S_{m,t} \}_{m=1}^{M}$, fronthaul/edge CSI as $G_{t}$/$H_{t}$ and the demand vector as $d_t$. 

 \noindent  \textit{Caching policy}:  In time slot $t$, the cached content at $m$th EN, $S_{m,t}$, is updated using the received fronthaul message $V_{m,t}$ and cached content of previous slot $S_{m,t-1}$. To meet the cache capacity requirement, we have $H(S_{m,t}) \leq \mu NL$. 

 \noindent \textit{Edge transmission policy}: The transmitted signal by EN $m$  at any time instant $j$ within a time slot $t$ is denoted by $X_{m,t,j}$ and obtained by the mapping $(d_t,H_{t},G_t,S_{m,t},V_{j',t}) \rightarrow X_{m,t,j}$ with $d_t$ as the demand vector, $G_{t}$/$H_{t}$ as fronthaul/edge CSI, $S_{m,t}$ as the cache content, $V_{j',t}$ as the fronthaul messages received at previous instants $j' \leq j-1$.



\noindent \textit{Long-term Normalized Delivery Time (NDT)}:
Denoting the transmission time of pipelined  strategy during slot $t$ as $T^{pl}_{t}$, the NDT at time slot $t$ is \cite{JSACAzimi} 
\begin{equation}\label{NDT_P}
\delta^{pl}_t(\mu,r)=\underset{L \rightarrow \infty}{\lim} \underset{P \rightarrow \infty}{\lim}\frac{\textrm{E}_{\mathcal{F}_t,H_{t},d_{t}}[T^{pl}_{t}]}{L/\log(P)}
\end{equation}
and, the \textit{long-term NDT} is defined as
\begin{equation}\label{avNDT_P}
\bar{\delta}^{pl}_{\text{on}}(\mu,r)=\underset{T \rightarrow \infty}{\lim \sup} \frac{1}{T}\underset{t=1}{\overset{T}{\sum}} \delta^{pl}_t(\mu,r).
\end{equation}
The minimum long-term NDT is denoted as $\bar{\delta}^{pl^*}_{\text{on}}(\mu,r)$. 

\section{Offline caching}\label{OFF}
In case of time-invariant popular set, namely $\mathcal{F}_t=\mathcal{F}$,  offline caching can be used with the following baseline delivery approaches \cite{Avik2}.

\noindent \textit{EN cooperation}: When $\mathcal{F}$ is cached in each EN, using joint Zero-Forcing (ZF) precoding, the resulting edge/fronthaul-NDTs\footnote{Edge/fronthaul NDTs have the same definition as \eqref{NDT_P} with the only caveat that they are defined for edge or fronthaul transmission \cite{JSACAzimi}.} are \cite{Avik2}
\begin{equation}\label{coop_e}
\delta_{\text{E,Coop}}=\frac{K}{\min\{M,K\}} ~\text{and}~\delta_{\text{F,Coop}}=0.
\end{equation}
\noindent \textit{EN coordination}: When non-overlapping fractions of each file is cached in each EN, fractions of requested files can be sent by each EN. Using interference alignment, we have 
\begin{equation}\label{coor_e}
\delta_{\text{E,Coor}}=\frac{M+K-1}{M}, ~\text{and}~\delta_{\text{F,Coor}}=0.
\end{equation}
\noindent \textit{C-RAN transmission}: In this mode, only cloud and fronthaul resources are used. In the worst case when distinct files are requested, multicasting $KL$ bits on the wireless fronthaul link results in $T_F=KL/(rlog(P))$. Then, ZF precoding at ENs is used. The resulting edge/fronthaul-NDTs are \cite{Avik2}
\begin{equation}\label{cran_e}
\delta_{\text{E,C-RAN}}(r)=\frac{K}{\min\{M,K\}} ~\text{and}~\delta_{\text{F,C-RAN}}(r)=\frac{K}{r}.
\end{equation}
\textit{Pipelined fronthaul-edge delivery} relies on simultaneous transmission on the fronthaul/edge channels. Denoting fronthaul/edge-NDTs by $\delta_{E}$/$\delta_{F}$, the NDT of pipelined scheme is \cite{Avik2}
\begin{equation}\label{def_pipe_off}
\delta^{pl}_{\text{off}}=\max \big ( \delta_{E},\delta_{F} \big ).
\end{equation}
The following propositions provide lower bound on the minimum offline NDT as well as achievable offline NDT. 
\begin{prop}\label{Lb_Off_Lem}
\noindent For $M \hspace{-1mm} \times \hspace{-1mm} K$ F-RAN with $N \hspace{-1mm} \geq \hspace{-1mm} K $, the minimum offline NDT is $\delta^{pl^*}_{\emph{off}}(\mu,r) \hspace{-1mm}  \geq  \hspace{-1mm} \delta^{pl}_{\emph{off,lb}}$ with
\begin{align}
\hspace{-2.5mm}\delta^{pl}_{\emph{off,lb}}  \hspace{-1mm} \triangleq \hspace{-3mm}  \underset{l \leq \min \{ M,K\}}     \max  \hspace{-1mm}  & \Bigg (\hspace{-1mm}\frac{K-(M-l)(K-l)\mu}{l + r}, 
 \frac{K}{\min (M,K)}\hspace{-1mm} \Bigg ).   \label{LB_OFF}
\end{align}
\end{prop} 
\begin{proof}
See Appendix \ref{Lb_Off_Proof}.
\end{proof}
\vspace{-3mm}
\begin{prop}\label{OfflinePipelinedNDT}
\noindent  For $M \times K$ F-RAN with  $N \geq K $ files, the achievable offline NDT satisfies 
\begin{align}
\hspace{-2.5mm} \delta^{pl}_{\emph{off,ach}}(\mu,r) =\delta_{\emph{off}}(\mu,r) \triangleq  \frac{K(1-\mu M)}{r},  \label{approach1}
\end{align}
for $\mu \in [0,\mu_1= K(1-r/\min(M,K))^+/(KM+r(\min(M,K)-1)) ]$, and
\begin{align}
\hspace{-2.5mm}\delta^{pl}_{\emph{off,ach}}(\mu,r) \hspace{-1mm}=\hspace{-1mm}   \big (\delta_{\emph{E,Coop}}\hspace{-.5mm}-\hspace{-.5mm}\delta_{\emph{off}}(\mu_1,r)  \hspace{-.5mm}\big )\hspace{-1mm} \Bigg (\hspace{-1mm} \frac{\mu-\mu_1}{\mu_2-\mu_1}\hspace{-1mm}\Bigg ) 
\hspace{-1.5mm}+\hspace{-.5mm} \delta_{\emph{off}}(\mu_1,r),  \label{approach2}
\end{align} 
for $\mu \in [\mu_1,\mu_2= \big ( 1-r/\min(M,K) \big ) ^+]$, and
\begin{align}
\delta^{pl}_{\emph{off,ach}}(\mu,r) =  \delta_{\emph{E,Coop}},  \label{approach3}
\end{align} 
for $\mu \in [\mu_2,1]$, and 
\begin{align}
\delta^{pl}_{\emph{off,ach}}(\mu,r) < 2 \delta^{pl}_{\emph{off,lb}},   \label{minoffline2} 
\end{align} 
with $\delta_{\emph{E,Coop}}$ and $\delta^{pl}_{\emph{off,lb}}$ defined in  \eqref{coop_e}  and Proposition \ref{Lb_Off_Lem}.
\end{prop}
\begin{proof}
Using per-block time sharing \cite{Avik2}, denoting achievable edge/fronthaul NDT of $\alpha$-fraction of file by $\delta_{E,1}$/$\delta_{F,1}$ and edge/fronthaul NDT of the remaining $(1-\alpha)$-fraction of file by $\delta_{E,2}$/$\delta_{F,2}$, time sharing achieves \eqref{def_pipe_off}
\begin{align}
\hspace{-2mm}\delta^{pl}_{\text{off,ach}} \hspace{-.5mm}=\hspace{-.5mm}  \max \Big \{\hspace{-.5mm} \alpha \delta_{F,1} \hspace{-.5mm}+\hspace{-.5mm} (1\hspace{-.5mm}-\hspace{-.5mm}\alpha) \delta_{F,2}, \alpha \delta_{E,1} \hspace{-.5mm}+\hspace{-.5mm} (1\hspace{-.5mm}-\hspace{-.5mm}\alpha) \delta_{E,2}\hspace{-.5mm} \Big \}.  \label{per_block_TS}
\end{align} 
\eqref{approach1} is obtained by plugging $\delta_{E,1}\hspace{-.5mm}=\hspace{-.5mm}\delta_{\text{E,Coor}}$, $\delta_{F,1}\hspace{-.5mm}=\hspace{-.5mm}\delta_{\text{F,Coor}}$, $\delta_{E,2}\hspace{-.5mm}=\hspace{-.5mm}\delta_{\text{E,C-RAN}}$ and $\delta_{F,2}\hspace{-.5mm}=\hspace{-.5mm}\delta_{\text{F,C-RAN}}$ into \eqref{per_block_TS} and setting $\alpha \hspace{-.5mm}=\hspace{-.5mm} \mu M$. Achievability of \eqref{approach2} follows from time-sharing between EN-coordination for $(\mu-\mu_1)/(\mu_2-\mu_1)$-fraction of each file and C-RAN transmission for the remaining  part. \eqref{approach3} is obtained  by plugging $\delta_{E,1}=\delta_{\text{E,Coop}}$, $\delta_{F,1}=\delta_{\text{F,Coop}}$, $\delta_{E,2}=\delta_{\text{E,C-RAN}}$ and $\delta_{F,2}=\delta_{\text{F,C-RAN}}$ into \eqref{per_block_TS} and setting $\alpha = \mu$. 
To prove \eqref{minoffline2}, setting $l=0$ in \eqref{LB_OFF} and comparing the result  either with \eqref{approach1}  for $\mu \in [0,\mu_1]$ or  with \eqref{approach3}  for $\mu \in [\mu_2,1]$ reveals that
\begin{equation}\label{Bnd1}
\delta^{pl}_{\text{off,ach}}(\mu,r)=\delta^{pl}_{\text{off,lb}}.
\end{equation}
For $\mu \in [\mu_1,\mu_2]$, we have 
\begin{align}
\hspace{-2mm}\frac{\delta^{pl}_{\text{off,ach}}(\mu,r)}{\delta^{pl}_{\text{off,lb}}}  \overset{(a)}  \leq  \Big ( \frac{K(1-\mu_1M)}{r} \Big )  \frac{\min \{ M,K\}}{K}  < 2,    \label{up2}
\end{align} 
where $(a)$ is obtained by using \eqref{approach1}, upper bounding the NDT  by setting $\mu=\mu_1 $ and using lower bound on the minimum NDT in Proposition \ref{Lb_Off_Lem}. Using \eqref{Bnd1}-\eqref{up2} results in \eqref{minoffline2}.  
\end{proof}
\begin{figure}[t]
\centering
\includegraphics[width=.4\textwidth]{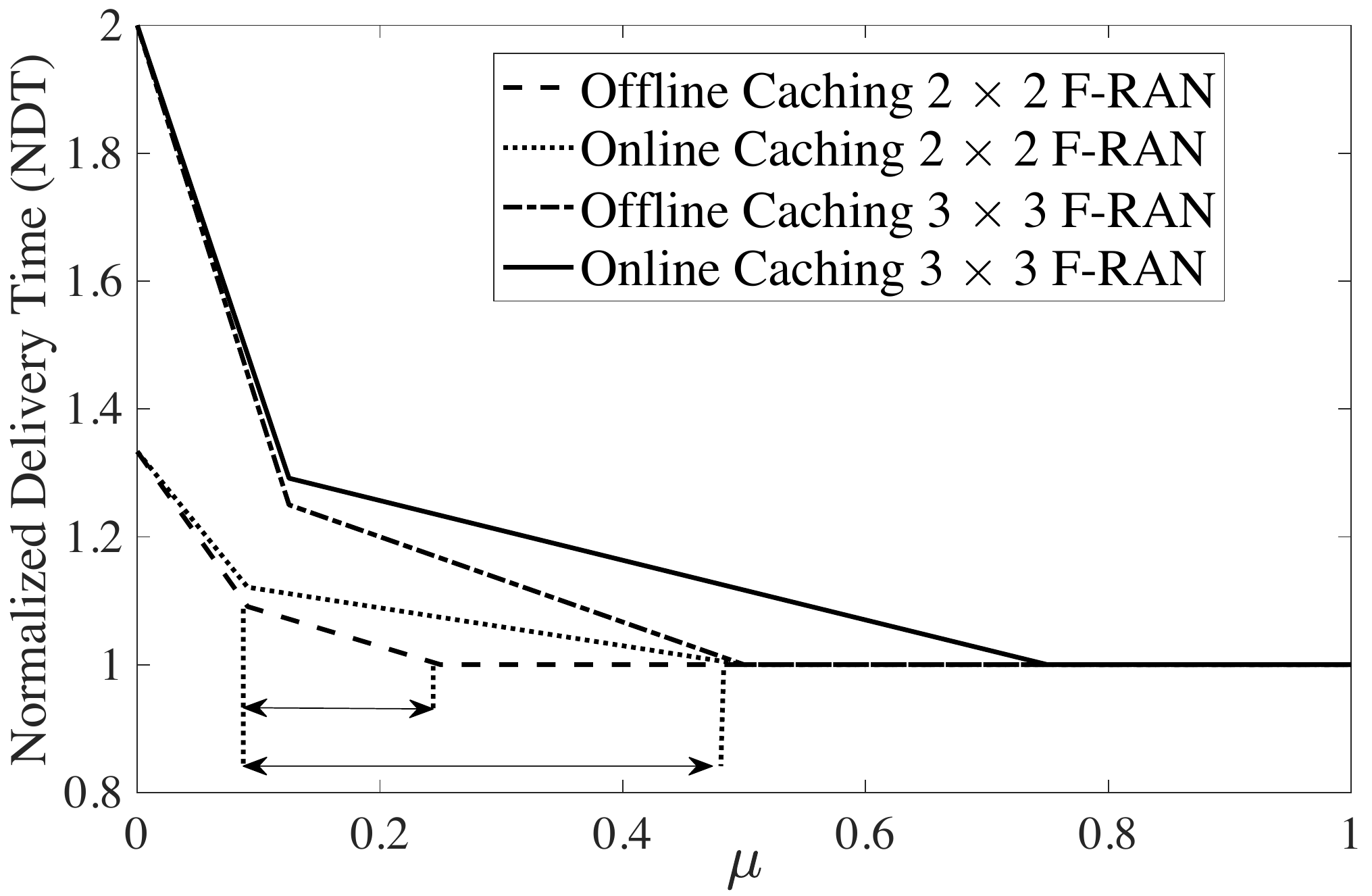}
\caption{ Normalized delivery time (NDT) as a function of fractional cache size $\mu$ for offline and online caching (with $p=1/2$) and power scaling $r=3/2$.}
\label{ref_off_on}
\vspace{-6mm}
\end{figure}
In Fig. \ref{ref_off_on}, the NDT of two offline F-RANs with $K=M=2$ and $K=M=3$ are shown. It is observed that by increasing $M$ and $K$, the \textit{intermediate cache regime} or equivalently $\mu_2-\mu_1$ will be increased. Also, it is required to have a larger $\mu_2$ to achieve \textit{full caching} feature. In case of full caching, for $\mu_2 \leq \mu \leq 1$, the NDT is the minimum possible obtained by $\mu=1$.
\section{Online caching}\label{sec_para}
In proactive online caching, upon the arrival of a newly popular file in the set, a $\mu$-fraction of it will be sent on the fronthaul link regardless of being requested by users or not.   In \cite{JSACAzimi}, proactive online caching is studied under dedicated fronthaul links while in this letter a wireless multicast fronthaul connection is considered.   
In what follows, a lower bound on the minimum long-term NDT is derived.
\begin{prop}\label{Lb_On_Coro}
\noindent  For $M \times K$ F-RAN with $N \geq K $, the minimum long-term NDT for online caching is 
\begin{align}
\bar{\delta}^{pl^*}_{\emph{on}}(\mu,r)  \geq \frac{(1-Kp/N)}{2} \delta^{pl^*}_{\emph{off}} (\mu,r)+ (Kp/N) \frac{M\mu}{r},   \label{PoC10}
\end{align}
with $\delta^{pl^*}_{\emph{off}} (\mu,r)$ defined in Proposition \ref{Lb_Off_Lem}.
\end{prop}
\begin{proof}
See Appendix \ref{Lb_On_Proof_Coro}.
\end{proof}
Proposition \ref{Lb_On_Coro} reveals that the lower bound scales inversely  with $r$, the power scaling of fronthaul transmission.
The next proposition provides the achievable long-term NDT. 
\begin{prop}\label{p_proact}
For $M \times K$ F-RAN with $N \geq K $ and pipelined transmission, proactive online caching achieves 
\begin{align}
\bar{\delta}^{pl}_{\emph{on,proact}}(\mu,r)=\delta_{\emph{on}}(\mu,r) \triangleq \frac{K(1-\mu M)}{r} + \frac{p\mu}{r},  \label{OnPL1}
\end{align}
for $\mu \in [0,\mu_1]$ with $\mu_1$ given in Proposition \ref{OfflinePipelinedNDT} and
\vspace{-2mm}
\begin{align}
\hspace{-1mm}\bar{\delta}^{pl}_{\emph{on,proact}}(\mu,r) \hspace{-1mm} = \hspace{-1mm} \Big (\hspace{-1mm}\delta_{\emph{E,Coop}}-\delta_{\emph{on}}(\mu_1,r)  \hspace{-1mm}\Big ) \Bigg ( \hspace{-1mm}\frac{\mu-\mu_1}{\mu'_2-\mu_1} \hspace{-1mm}\Bigg )\hspace{-.5mm} + \hspace{-.5mm} \delta_{\emph{on}}(\mu_1,r) ,  \label{OnPL2}
\end{align} 
\vspace{-2mm}
for $\mu \in [\mu_1,\mu'_2= \Big (K( 1-r/\min(M,K)))/(K-1) \Big ) ^+]$, and
\begin{align}
\bar{\delta}^{pl}_{\emph{on,proact}}(\mu,r) = \delta_{\emph{E,Coop}}     \label{OnPL3}
\end{align} 
for $\mu \in [\mu'_2,1]$.
\end{prop}
\begin{proof}
\begin{figure}[t]
\centering
\includegraphics[width=.44\textwidth]{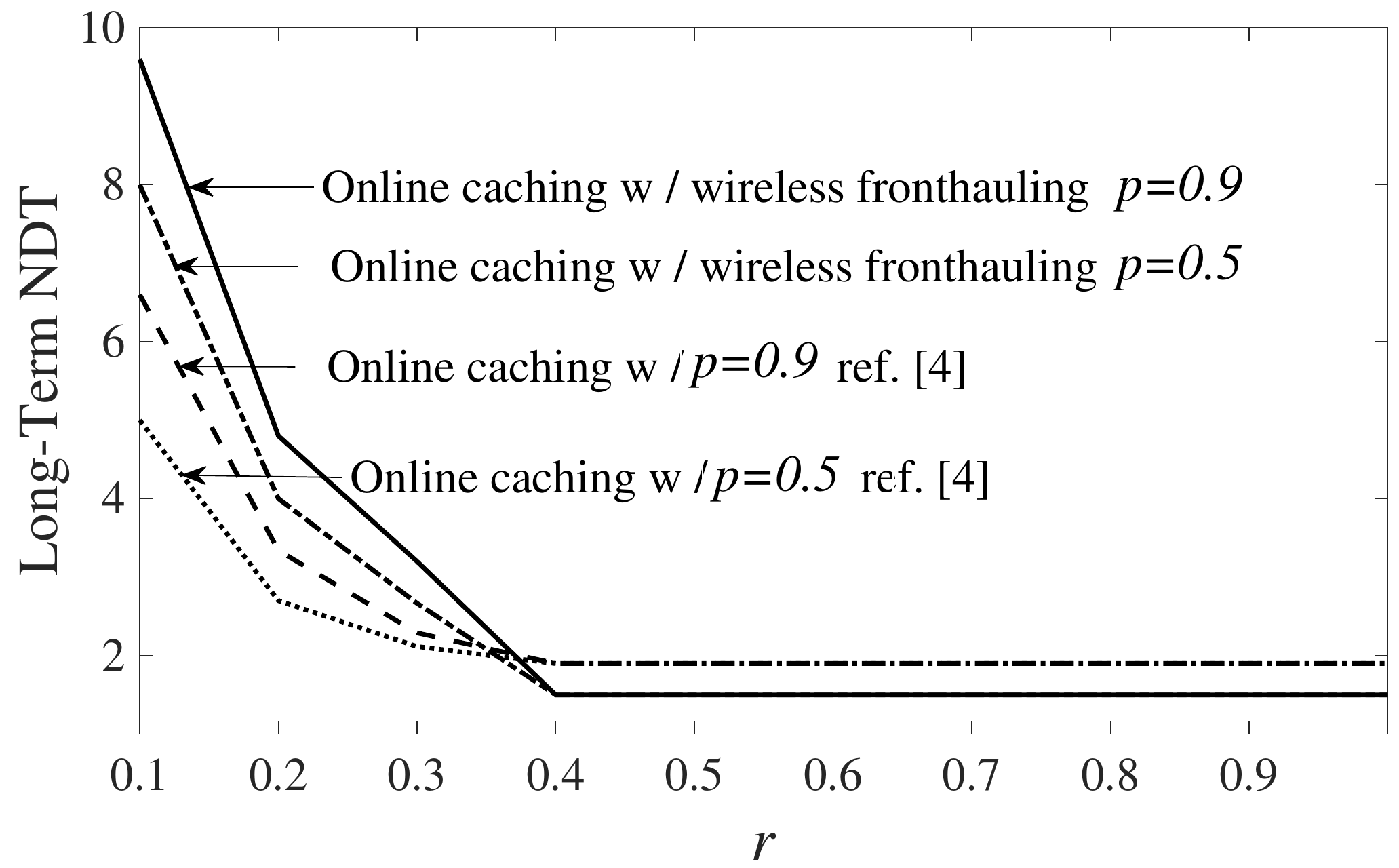}
\vspace{-2mm}
\caption{ Long-term NDT as a function of power scaling of fronthaul transmission or fronthaul rate $r$ for online caching with $p=0.5$, $p=0.9$, $\mu=0.4$, $M=2$, $K=3$.}
\label{ref_on_r}
\end{figure}
The arrival of a newly popular file in the set, which occurs with probability of $p$, requires proactive transmission on the wireless fronthaul channel and hence increasing the fronthaul-NDT of offline scheme by the term $\mu/r$. Instead, with  probability of $(1-p)$, the popular set will remain the same as previous slot  and the NDT  is obtained similar to  Proposition \ref{OfflinePipelinedNDT}. Using per-block time sharing, we have
\begin{align}
\bar{\delta}^{pl}_{\text{on,proact}}(\mu,r) & =  p \max \Big \{ \alpha \delta_{F,1} + (1-\alpha) \delta_{F,2}+ \frac{\mu}{r}, \nonumber \\
& \alpha \delta_{E,1} + (1-\alpha) \delta_{E,2} \Big \}+ (1-p) \delta^{pl}_{\text{off,ach}},
\label{per_block_TS_On}
\end{align}
with $ \delta^{pl}_{\text{off,ach}}$ given in \eqref{per_block_TS}. To complete the proof, $\delta_{E,1}$, $\delta_{F,1}$, $\delta_{E,2}$ and $\delta_{F,2}$ are plugged into \eqref{per_block_TS_On} in the same way as described in the proof of Proposition \ref{OfflinePipelinedNDT}.
\end{proof}
In Fig. \ref{ref_off_on}, time-varying content popularity with probability $p=0.5$ results in higher NDT for online caching comparing to offline caching. Furthermore, a larger cache storage is required to have full caching. 
The next proposition relates the long-term NDT of online caching with minimum NDT of offline caching.  
\begin{prop}\label{final_uppro1}
\noindent For $M \times K$ F-RAN with $N \geq K $, $0 < r < \min(M,K)$ and pipelined transmission, the minimum long-term NDT of online caching $\bar{\delta}_{\emph{on}}^{pl^*}(\mu,r)$ satisfies the condition
\begin{equation}\label{longpro1}
\bar{\delta}_{\emph{on}}^{pl^*}(\mu,r) = 2\delta^{pl^*}_{\emph{off}} (\mu,r) +  O \Big (\frac{1}{r} \Big ),
\end{equation}
with $\delta^{pl^*}_{\emph{off}} (\mu,r)$ defined in Proposition \ref{Lb_Off_Lem}.  
\end{prop}
\begin{proof}
It is sufficient to prove that the following inequalities.
\begin{align}
 \frac{1-\frac{Kp}{N}}{2}\delta^{pl^*}_{\text{off}} (\mu,r) & + \Big ( \frac{Kp}{N} \Big ) \frac{M \mu}{r}  \overset{(a)} \leq  \bar{\delta}_{\text{on}}^{pl^*}(\mu,r) \overset{(b)} \leq  \bar{\delta}_{\text{on,proact}}^{pl}(\mu,r)  \nonumber \\  
 & \overset{(c)} \leq  2 \delta^{pl^*}_{\text{off}} (\mu,r) +  6+\frac{4p}{r}. \label{longpro2}
\end{align}
$(a)$ follows from Proposition \ref{Lb_On_Coro}, $(b)$ holds by definition, $(c)$ is proved  as follows.  

\noindent $(i)$ For $\mu\in [0,\mu_1]$,  using \eqref{approach1} and \eqref{OnPL1}, we have:
\begin{align}
\bar{\delta}^{pl}_{\text{on,proact}}(\mu,r) =\delta^{pl}_{\text{off,ach}}(\mu,r) + \frac{p\mu}{r}, \label{Rel1}
\end{align}

\noindent $(ii)$ For $\mu\in [\mu_1,\mu_2]$, using \eqref{approach2} and \eqref{OnPL2} and then noting the the fact that $\mu'_2 \geq \mu_2$, we have
\begin{align}
\bar{\delta}^{pl}_{\text{on,proact}}(\mu,r)&-\delta^{pl}_{\text{off,ach}}(\mu,r)   \leq  \frac{p\mu}{r} \label{Rel2}
\end{align}  

\noindent $(iii)$ For $\mu\in [\mu_2,\mu'_2]$, the following inequality holds:
\begin{align} 
&\bar{\delta}^{pl}_{\text{on,proact}}(\mu,r)-\delta^{pl}_{\text{off,ach}}(\mu,r)  
 \overset{(a)} \leq  \Bigg (\hspace{-1mm}\frac{K(1-\mu_1 M)}{r} + \frac{p \mu_1}{r}\hspace{-1mm}\Bigg ) \times \nonumber \\ 
& \Bigg (\hspace{-1mm} 1\hspace{-.5mm}-\hspace{-.5mm}\frac{\mu_2-\mu_1}{\mu'_2-\mu_1} \hspace{-1mm}\Bigg )\hspace{-1mm}  
 \overset{(b)} \leq \hspace{-1mm} \Bigg (\hspace{-1mm}1\hspace{-.5mm}+\hspace{-.5mm}\frac{K}{M} \hspace{-.5mm}+\hspace{-.5mm} \frac{p}{r}\hspace{-1mm}\Bigg ) \hspace{-1mm}\Bigg (\hspace{-1mm} \frac{M}{K(M-1)} \hspace{-1mm}\Bigg ) \hspace{-1mm} \overset{(c)} \leq \hspace{-1mm} 3 \hspace{-.5mm}+\hspace{-.5mm} \frac{2p }{r}, \label{Rel3}
\end{align}
where $(a)$ is obtained using \eqref{approach3} and \eqref{OnPL2}, dropping the negative term and then setting $\mu=\mu_2$ to maximize the upper bound, $(b)$ is obtained by using the definition of $\mu_1$, $\mu_2$ and $\mu'_2$ in Proposition \ref{OfflinePipelinedNDT} and Proposition \ref{p_proact} and $(c)$ is obtained by maximizing the bound using $M=2$ and $K=1$. 

\noindent $(iv)$ For $\mu\in [\mu'_2,1]$, using \eqref{approach3} and \eqref{OnPL3}, results in:
\begin{align}
\bar{\delta}^{pl}_{\text{on,proact}}(\mu,r)=\delta^{pl}_{\text{off,ach}}(\mu,r)  = \frac{K}{\min (M,K)}. \label{Rel4}
\end{align} 
Using \eqref{Rel1}-\eqref{Rel4} completes the proof of \eqref{longpro2}$-(c)$.
\end{proof}
From Proposition \ref{final_uppro1}, it can be inferred that the time-variability of popular set results in an additional cost on the long-term NDT that increases inversely with respect to the power scaling of fronthaul transmission $r$. Time varying  content popularity make it inevitable to deliver the newly popular files only by using fronthaul resources. In Fig. \ref{ref_on_r}, it is observed that long-term NDT is a decreasing function of $r$. Comparing to dedicated fronthaul links in \cite{JSACAzimi}, wireless multicast fronthauling has higher long-term NDT in low fronthaul regime while after a threshold multicasting with pipelined transmission outperform dedicated approach. Since there is not a closed formula for long-term NDT in \cite{JSACAzimi}, the threshold cannot be derived analytically.
\section{Conclusions} \label{Conc}
In this letter, the delivery of  time-varying content in a fog network with wireless fronthaul and edge connection is considered. To this end, proactive online caching with pipelined transmission is introduced. Using information-theoretic analysis, it is proved that performance of the system is a function of fronthaul and edge resources. Since the only means of delivering new content is through the wireless multicast fronthaul links, the power scaling of fronthaul transmission is a fundamental limit on the achievable latency. In future work, we will extend the model to the case with imperfect CSI.
\appendices
\section{Proof of Proposition \ref{Lb_Off_Lem}}\label{Lb_Off_Proof}
Assuming $N \hspace{-.5mm} \geq \hspace{-.5mm} K$, let $F_{[1:K]}$ be the vector of requested files and $T$ denotes the delivery latency. Subsets of information resources are considered such that each subset is sufficient to decode the files at ENs or the user side. The first subset is comprised of  $i)$ The $0 \hspace{-.5mm}\leq \hspace{-.5mm} l \hspace{-.5mm}\leq \hspace{-.5mm} \min (M,K)$ outputs of wireless edge channel, without loss of generality the output of channel at the first user to
the $l^{th}$ user is considered, namely $Y^T_{[1:l]}$, $ii)$ The cached contents at $(M-l)$ ENs $S_{[1:(M-l)]}$ , $iii)$ The $(M-l)$ outputs of wireless multicast fronthaul channel  $V_{[1:(M-l)]}^T$. We have
\vspace{-3mm} 
\begin{align}
&KL=H(F_{[1:K]})   \hspace{-1mm} \overset{(a)} \leq \hspace{-1mm} lT \log(P) \hspace{-1mm} + \hspace{-1mm} L \epsilon_L \hspace{-1mm}+\hspace{-3mm} \sum_{i=1}^{(M-l)}\hspace{-2mm} H(S_{i}| F_{[1:l] \cup [K+1:N]}) \hspace{-1mm} \nonumber  \\ 
&+ \hspace{-1mm}H(G_{[1,M-l]}U^T\hspace{-1mm}+\hspace{-1mm}W_{[1,M-l]}^T| F_{[1:l] \cup [K+1:N]})  \nonumber  \\
&\hspace{-1mm} \overset{(b)} \leq (l+r) T \log(P) \hspace{-1mm}+\hspace{-1mm} L \epsilon_L \hspace{-1mm} +\hspace{-1mm} T \epsilon_P \hspace{-1mm}+\hspace{-1mm} (M-l)(K-l)\mu L,   \label{PoC1}
\end{align}  
where $(a)$ is obtained by using first \cite[eq. 64-(a)]{Avik2}, Fano's inequality where $\epsilon_L \rightarrow 0$ as $L \rightarrow \infty$ and using \cite[Lemma 5]{Avik2} and $(b)$ is obtained using $H(S_{i}) \leq \mu NL$ and $\epsilon_P/ \log P \rightarrow 0$ as $P \rightarrow \infty$. Rearranging \eqref{PoC1} and taking the limits results in the minimum NDT
\begin{align}
\delta^{pl^*}_{\text{off}}(\mu,r)  \geq  \frac{K-(M-l)(K-l)\mu}{l + r},    \label{PoC3}
\end{align}   
for $l \leq \min(M,K)$.
The second subset is $K$ received signals by $K$ users i.e., $Y^T_{[1:K]}$. Following the same argument to prove \cite[eq. 69]{Avik2} and the equation after that, we have 
\begin{align}
\delta^{pl^*}_{\text{off}}(\mu,r)  \geq \frac{K}{\min (M,K)}.    \label{PoC9}
\end{align}
Using inequalities \eqref{PoC3} and \eqref{PoC9}, completes the proof. 
\section{Proof of Proposition \ref{Lb_On_Coro}}\label{Lb_On_Proof_Coro}
If a genie provides the side information to the ENs about the possible new file in  $\mathcal{F}_t$ and noting that this file will be requested by users with the probability of $Kp/N$, we have 
\begin{align}
\delta^{pl}_{\text{t}}(\mu,r)  \geq (1-Kp/N)\delta^{pl}_{\text{off,lb}} + (Kp/N) \delta^{pl}_{\text{on,lb}},   \label{NDT_t}
\end{align}
where $\delta^{pl}_{\text{off,lb}}$ is lower bound on NDT of the offline caching given in Proposition \ref{Lb_Off_Lem} and $\delta^{pl}_{\text{on,lb}}$ is the lower bound on the minimum NDT of online caching. Furthermore, we have:
\begin{align}
\frac{\delta^{pl^*}_{\text{off}}(\mu,r) }{\delta^{pl}_{\text{off,lb}}} \overset{(a)} \leq \frac{\delta^{pl}_{\text{off,ach}}(\mu,r) }{\delta^{pl}_{\text{off,b}}} \overset{(b)} < 2 ,  \label{ratios} 
\end{align} 
where $(a)$ is obtained using the fact that the optimum policy minimizes the achievable NDT and $(b)$ is obtained using \eqref{minoffline2}.  By plugging \eqref{ratios} into  \eqref{NDT_t} and then using \eqref{avNDT_P}, we have
\begin{align}
\bar{\delta}^{pl^*}_{\text{on}}(\mu,r)  \geq \frac{(1-Kp/N)}{2}\delta^{pl^*}_{\text{off}}(\mu,r) + (Kp/N) \delta^{pl}_{\text{on,lb}}.   \label{PoC10}
\end{align} 
Next, we have
\begin{align}
\hspace{-3mm} \delta^{pl}_{\text{on,lb}} & \overset{(a)} \geq \frac{K-(M-l)(K-l-1)\mu}{l + r} \nonumber \\ 
\hspace{-3mm} &= \frac{K\big (1- \mu (M-l) \big )+(M-l)(l+1) \mu}{r} \overset{(b)} \geq \frac{M \mu}{r}, \label{coro_b}
\end{align}
\vspace{-1mm}
where $(a)$ is obtained similar to \eqref{PoC1} with only one caveat that in \eqref{PoC1}-$(e)$, $\sum_{i=1}^{(M-l)} H(S_{i}| F_{[1:l] \cup [K+1:N]}) \leq (M-l)(K-l-1)\mu L $ since one of the requested files is new and not stored at the ENs' caches and $(b)$ is obtained by dropping the positive term since $0 \leq \mu \leq 1/(M-l)$ and then setting $l=0$. Plugging \eqref{coro_b} into \eqref{PoC10}, completes the proof.


\ifCLASSOPTIONcaptionsoff
  \newpage
\fi



%
\bibliographystyle{IEEEtran}
\bibliography{IEEEabrv,IEEEexample}

%




\end{document}